\documentclass{llncs}
\usepackage{}
\usepackage{amsmath}
\usepackage{amsfonts}
\usepackage{amssymb}
\usepackage{makeidx}  
\usepackage{amsfonts,amssymb}
\usepackage{graphicx}
\begin{document}

\title{Strengthened Hardness for Approximating Minimum Unique Game and Small Set Expansion}
\titlerunning{Strengthened Hardness}  
\renewcommand\thefootnote{}

\author{Peng Cui\inst{1}}

\authorrunning{Peng Cui}   
%
\tocauthor{Peng Cui}
\institute{Key Laboratory of Data Engineering and Knowledge Engineering, MOE,
School of Information Resource Management, Renmin University of China, Beijing
100872, P. R. China.\\
\email{cuipeng@ruc.edu.cn}}

\maketitle              

\begin{abstract}
In this paper, the author puts forward a variation of Feige's Hypothesis, which claims that it is hard on average refuting Unbalanced Max 3-XOR under biased assignments on a natural distribution. Under this hypothesis, the author strengthens the previous known hardness for approximating Minimum Unique Game, $5/4-\epsilon$, by proving that Min 2-Lin-2 is hard to within $3/2-\epsilon$ and strengthens the previous known hardness for approximating Small Set Expansion, $4/3-\epsilon$, by proving that Min Bisection is hard to approximate within $3-\epsilon$. In addition, the author discusses the limitation of this method to show that it can strengthen the hardness for approximating Minimum Unique Game to $2-\kappa$ where $\kappa$ is a small absolute positive, but is short of proving $\omega_k(1)$ hardness for Minimum Unique Game (or Small Set Expansion), by assuming a generalization of this hypothesis on Unbalanced Max k-CSP with Samorodnitsky-Trevisan hypergraph predicate.
\end{abstract}

\section{Introduction}
Recently, the authors of \cite{ow1} show a new point of $(1/2,3/8)$-approximation NP-hardness of UG, which is an improvement of previously known $(3/4,11/16)$-approximation NP-hardness of UG\cite{h}. Their result determines a two-dimensional region of for $(c,s)$-approximation NP-hardness of UG, namely, the triangle with the three vertices $(0,0)$, $(1/2,3/8)$, and $(1,1)$. All known points of $(c,s)$-approximation NP-hardness of UG are in the triangle, plus an inferior bump area near the origin by \cite{fr} (See Fig. 1).

The curve $(1-\epsilon,1-c\sqrt{\epsilon})$ for some constant $c$ and small $\epsilon$ seems a separate line between hard and easy regions of the $(c,s)$ plane. By Raz's parallel repetition theorem\cite{r}, UGC is implied by $(\epsilon,\epsilon^p)$-approximation NP-hardness of Min UGC for any constant $1/2<p<1$. In the opposite direction, UGC implies $(\epsilon,c_1\sqrt{\epsilon})$-approximation NP-hardness of Min 2-Lin-2 for some constant $c_1$\cite{kkmo}. In the hardness side, the standard Max Cut SDP relaxation has a SDP gap $(\epsilon,c_2\sqrt{\epsilon})$ for some constant $c_2$ with respect to Min UnCut\cite{ow2}. In the algorithm side, the subexponential algorithm given by \cite{abs} returns a solution with value at most $c_3\sqrt{\epsilon}$ on instance of Min UG with value $\epsilon$ for some constant $c_3$. However, the best known NP-hardness for approximating Min UG is still $5/4-\epsilon$ despite all the efforts.

It is known that we can rule out the possibility of PTAS for Min Bisection under a complexity assumption stronger than $NP\ne P$\cite{k2}, and that Min Bisection is hard to approximate within $4/3-\epsilon$ assuming Feige's Hypothesis\cite{f}. It would be interesting to answer the question whether we can further enlarge the hardness gap of Min Bisection (and SSE).

The authors of \cite{f,a} have established connection between approximation complexity and average complexity. They use average complexity to prove inapproximability results for some famous problems, which has resisted discovery of meaningful inapproximability results under standard complexity assumptions. A recent example of such problems is Densest $\kappa$-Subgraph\cite{aammw}.

In this paper, the author puts forward a variation of Feige's Hypothesis, which claims that it is hard on average refuting Unbalanced Max 3-XOR under biased assignments on a natural distribution. Under this hypothesis, the author strengthens the previous known hardness for approximating Minimum Unique Game, $5/4-\epsilon$, by proving that Min 2-Lin-2 is hard to within $3/2-\epsilon$ and strengthens the previous known hardness for approximating Small Set Expansion, $4/3-\epsilon$, by proving that Min Bisection is hard to approximate within $3-\epsilon$.

In addition, the author discusses the limitation of this method to show that it can strengthen the hardness for approximating Minimum Unique Game to $2-\kappa$ where $\kappa$ is a small absolute positive, but is short of proving $\omega_k(1)$ hardness for Minimum Unique Game (or Small Set Expansion), by assuming a generalization of this hypothesis on Unbalanced Max k-CSP with Samorodnitsky-Trevisan hypergraph predicate.

\section{Preliminaies}
In {\it Unique Game (UG)}, we are given a graph $G=(V,E)$, and a set of labels, $[k]$. Each edge $e=(u,v)$ in the graph is equipped with a permutation $\pi_{e}:[k]\rightarrow[k]$. The solution of the problem is a labeling $f:V\rightarrow [k]$ that assigns a label to each vertex of $G$. An edge $e=(u,v)$ is said to be satisfied under $f$ if $\pi_{e}(f(u))=f(v)$. The goal of the problem is to find a labeling such that the number of the satisfied edges under this labeling is maximized. The value of the instance $Val(I)$ is defined as the maximum fraction of the satisfied edges over all labeling. In the same situation of UG, the goal of {\it Minimum Unique Game (Min UG)} is to find a labeling such that the number of the unsatisfied edges under this labeling is minimized. The value of the instance $Val(I)$ is defined as the minimum fraction of the unsatisfied edges over all labeling.

In {\it Max 2-Lin-2}, we are given a set of linear equations over $GF[2]$. Each equation contains exactly two variables. The goal of the problem is to seek an assignment of the variables such that the number of satisfied equations is maximized. In the same situation of Max 2-Lin-2, the goal of {\it Min 2-Lin-2} is to seek an assignment of the variables such that the number of unsatisfied equations is minimized. In {\it Max Cut}, we are given a graph, and the goal of the problem is to seek a cut of the graph with maximum edges. In {\it Min UnCut}, we are given a graph, and the goal of the problem is to seek a cut of the graph that leaves minimum edges uncut. Note that Max 2-Lin-2 and Max Cut are two special cases of UG, and Min 2-Lin-2 and Min UnCut are two special cases of Min UG.

The {\it $(c,s)$-approximation NP-hardness} of UG is defined as: for some fixed $0<s<c<1$, there is a $k$ such that given an instance $I$ of UG with $k$ labels it is NP-hard to distinguish whether $Val(I)\ge c$ or $Val(I)<s+\epsilon$ for any $\epsilon>0$. For any fixed $0<c'<s'<1$, the {\it $(c',s')$-approximation NP-hardness} of Min UG is defined as: there is a $k$ such that given an instance $I$ of Min UG with $k$ labels it is NP-hard to distinguish whether $Val(I)>s'-\epsilon$ or $Val(I)<c'+\epsilon$ for any $\epsilon>0$. Similarly, we can define the {\it $(c',s')$-approximation hardness} of Min UG under Conjecture 2 or Conjecture 4.

In {\it Small Set Expansion (SSE)}, we are given a graph $G=(V,E)$ and a constant $0<\delta\le 1/2$. The goal of the problem is to find a subset $S\subseteq V$ satisfying $|S|/|V|=\delta$ such that $\Phi(S)$, the edge expansion of $S$ is minimized. The {\it edge expansion} $\Phi(S)$ of a subset $S\subseteq V$ is defined as: $\Phi(S)=\frac{|V||E(S,V\setminus S)|)}{|E||S|}$. The {\it expansion profile} is defined as: $\Phi_{G}(\delta)=\min_{|S|/|V|=\delta}{\Phi(S)}$, where $0<\delta\le 1/2$. As a special case of Small Set Expansion Problem, {\it Min Bisection} is defined as: given a graph $G$ with $n$ vertices, where $n$ is even, find a set $S$ of $n/2$ vertices (a bisection) such that the number of edges connecting $S$ and $V\setminus S$ (the bisection width) is minimized.

The {\it Unique Game Conjecture (UGC)}\cite{k1} states: for every $\zeta,\delta>0$, there is a $k=k(\zeta,\delta)$ such that given an instance $I$ of UG with $k$ labels it is NP-hard to distinguish whether $Val(I)>1-\zeta$ or $Val(I)<\delta$. The {\it Small Set Expansion Hypothesis (SSEH)}\cite{rs} states: for every $\eta>0$, there is a $\delta$ such that it is NP-hard to distinguish whether $\Phi_{G}(\delta)>1-\eta$ or $\Phi_{G}(\delta)<\eta$. The authors of \cite{rs} show that SSEH implies UGC.

Throughout this paper, let $\beta\diamond\gamma=\beta+\gamma-2\beta\gamma$, and $\epsilon$ generally denotes a negligible quantity.

\section{Conjectures on Unbalanced 3-XOR and 3-AND}
In this section, the author puts forward a variation of Feige's Hypothesis, which claims it is hard on average refuting Unbalanced Max 3-XOR under biased assignments on a natural distribution. We can
strengthen the previous known hardness for approximating Minimum Unique Game, $5/4-\epsilon$, by proving that Min 2Lin-2 is hard to approximate within $3/2-\epsilon$.

In {\it Max 3-XOR}, we are given a set of $XOR$ clauses, each clause contains exactly three literals. The goal of the problem is to seek an assignment of the Boolean variables such that the number of satisfied clauses is maximized. In {\it Max 3-AND}, we are given a set of $AND$ clauses, each clause contains exactly three literals. The goal of the problem is to seek an assignment of the Boolean variables such that the number of satisfied clauses is maximized.

In {\it Random Unbalanced Max 3-XOR}, we assume that formulas are generated by the following random process. Given parameters $n$ and $m$, each clause is generated independently at random by selecting the three variables in it independently at random and inserting the negative literal of the variable into the clause with probability $\beta<1/2$ and inserting the positive literal of the variable into the clause with probability $1-\beta$. $\beta$ is called {\it imbalance} of the instance, and the instance is called $\beta$-balanced. In addition, We are interested in the assignments such that the fraction of variables assigned to 0 is no more than $\gamma$, which is called {\it bias} of the assignments. In {\it Random Unbalanced Max 3-AND}, formulas are generated similarly, and we can define the notations, {\it imbalance} and {\it bias}, similarly.

In this paper, the author considers the average complexity of Random Unbalanced Max 3-XOR, and put forward a variation of Feige's Hypothesis\cite{f,a}.\\

\noindent{\bf Conjecture 1. }
{\it For every $0<\gamma<\beta<1/2$, for every fixed $\epsilon>0$, for $\Delta$ a sufficiently large constant independent of $n$, there is no polynomial time algorithm that refutes most $\beta$-balanced Max 3-XOR formulas with $n$ variables and $m=\Delta n$ clauses, but never refutes a $1-\epsilon$ satisfiable formula under $\gamma$-biased assignments.}\\

The author also considers the average complexity of Random Unbalanced Max 3-AND, and put forward the following conjecture.\\

\noindent{\bf Conjecture 2. }
{\it For every $0<\gamma<\beta<1/2$, for every fixed $\epsilon>0$, for $\Delta$ a sufficiently large constant independent of $n$, there is no polynomial time algorithm that refutes most $\beta$-balanced Max 3-AND formulas with $n$ variables and $m=\Delta n$ clauses, but never refutes a $1-\frac{3}{2}\beta\diamond\gamma-\epsilon$ satisfiable formula under $\gamma$-biased assignments.}

\begin{theorem}
Conjecture 1 implies Conjecture 2.
\end{theorem}
\begin{proof}
We rewrite a formula of $\beta$-balanced Max 3-XOR to a formula of $\beta$-balanced Max 3-AND. If the formula of Max 3-XOR is random, then the formula of Max 3-AND is also random. If the formula of Max 3-XOR $\phi$ is $1-\epsilon$ satisfiable by $\gamma$-biased assignments, we show in the following that at least $1-\frac{3}{2}\beta\diamond\gamma-\epsilon$ fraction of clauses in $\phi$ have all the three literals satisfied.

On average, each positive literal has $3(1-\beta)\Delta$ appearance in $\phi$, and each negative literal has $3\beta\Delta$ appearance in $\phi$. When $\Delta$ is large enough, standard bounds on large deviations show that with high probability, all but an $\epsilon$ fraction of the occurrences of positive literals correspond to positive literals that appear between $(3(1-\beta)\pm\epsilon)\Delta$ times in $\phi$, and all but an $\epsilon$ fraction of the occurrences of negative literals correspond to negative literals that appear between $(3\beta\pm\epsilon)\Delta$ times in $\phi$.

If this does hold, observe that every $\gamma$-biased assignment $\psi$ does not satisfy on average at most $3(\beta(1-\gamma)+\gamma(1-\beta))+\epsilon$ variables per clause in $\phi$. It then follows that at most $\frac{3}{2}(\beta(1-\gamma)+\gamma(1-\beta))+\epsilon$ clauses have exactly one literal satisfied by $\psi$.
\end{proof}

\begin{theorem}
Conjecture 2 holds for any $0<\gamma<\beta<1/2$ implies $(c',s')$-approximation hardness of Min UG for $c'=\textstyle\frac{1}{2}\beta\diamond\gamma$ and $s'=\textstyle\frac{1}{4}(1-(1-\beta)^3)-\epsilon$.
\end{theorem}
\begin{proof}
We use the three-dimensional cube gadget that is similar to the gadgets used by authors of \cite{h}.

Let $l_{1}\wedge l_{2}\wedge  l_{3}$ be a clause in the formula of Max 3-AND, where $l_{i}$ is either a variable $x_{i}$ or its negation $\bar x_{i}$, for $i=1,2,3$. The set of equations we construct have variables at the corners of a three-dimensional cube, which take value $1$ or $-1$. For each $\mu\in\{0,1\}^3$, we have a variable $v_{\mu}$. The variable $v_{000}$ is replaced by $w$ taking value $-1$. We let $u_{1}$ take the place of $v_{011}$, $u_{2}$ the place of $v_{101}$, and $u_{3}$ the place of $v_{110}$, where $u_{i}=-1$ if $x_{i}=1$, and $u_{i}=1$ if $x_{i}=0$. For each edge $(w,v_{\mu})$ of the cube, we have the equation $wv_{\mu}=-1$. For each edge $(u_{i},v_{\mu})$ of the cube, we have the equation $u_{i}v_{\mu}=1$ if $l_{i}$ is positive, and the equation $u_{i}v_{\mu}=-1$ if $l_{i}$ is negative, for all $i=1,2,3$.

If all $l_{i}$ are satisfied in the clause, we assign $v_{\mu}$ the value $(-1)^{\mu_{1}+\mu_{2}+\mu_{3}}$. All the twelve edge equation are satisfied and left no equation unsatisfied. Otherwise, an enumeration establishes that it is only possible to satisfy at most nine equations and left three equations unsatisfied, and that it is always possible to satisfy at least eight equations and left four equations unsatisfied.

Given a $\beta$-balanced Max 3-AND formulas $\phi$ that is at most $(1-\beta)^3+\epsilon$ satisfiable, at least $1-(1-\beta)^3-\epsilon$ clauses in $\phi$ are unsatisfied.

Now we reduce a formula of $\beta$-balanced Max 3-AND to an instance of Min 2-Lin-2 using the gadget introduced above. If the formula is $1-\frac{3}{2}\beta\diamond\gamma-\epsilon$ satisfiable under $\gamma$-biased assignments, then the value of the instance of Min 2-Lin-2 is at most $$\textstyle\frac{3}{2}\beta\diamond\gamma+\epsilon-(\textstyle\frac{3}{2}\beta\diamond\gamma+\epsilon)\textstyle\frac{2}{3}=\textstyle\frac{1}{2}\beta\diamond\gamma+\epsilon.$$ If the formula is random, then it is at most $(1-\beta)^3+\epsilon$ satisfiable in high probability, which implies the value of the instance of Min 2-Lin-2 is at least $$1-(1-\beta)^3-\epsilon-(1-(1-\beta)^3-\epsilon)\textstyle\frac{3}{4}=\textstyle\frac{1}{4}(1-(1-\beta)^3)-\epsilon.$$
\end{proof}

\begin{corollary}
Conjecture 2 holds for arbitrarily small $\beta$ and $\gamma$ implies Min UG is hard to approximate within $3/2-\epsilon$.
\end{corollary}

\begin{lemma}
For an integer $k\ge 3$ and every $\epsilon>0$, there is some $\Delta_{\epsilon}>0$ such that for every $\Delta>\Delta_{\epsilon}$, $n$ large enough, and $0<\gamma<\beta<1/2$, with high probability the following holds. Every set of $((1-\beta\diamond\gamma)^k+\epsilon)m$ clauses in a random $\beta$-balanced Max 3-AND formula with $m=\Delta n$ clauses contains at least $\gamma n+1$ different negative literals or $(1-\gamma)n+1$ different positive literals.
\end{lemma}
\begin{proof}
Fix a set $S$ of $n$ literals with exactly $\gamma$ fraction of positive literals to be avoided. The probability that a random clause with three literals avoids these literals is $(1-\beta\diamond\gamma)^3$. For large enough $\Delta$, standard bounds on large deviations implies that with probability greater than $1-(1-\beta\diamond\gamma)^{3n}$, less than $((1-\beta\diamond\gamma)^k+\epsilon)m$ random clauses avoid the set $S$. As there are roughly $2^{2\gamma n}$ ways of choosing the set $S$, the union bound implies that on one of them is avoided by a set of $((1-\beta\diamond\gamma)^k+\epsilon)m$ clauses.
\end{proof}

\begin{theorem}
Conjecture 2 holds for any $0<\gamma<\beta<1/2$ implies Small Set Expansion is hard to approximate within $\frac{2(1-(1-\beta)^3)}{\frac{3}{2}\beta\diamond\gamma}-1-\epsilon$.
\end{theorem}
\begin{proof}
We reduce $\beta$-balanced Max 3-AND to Min Bisection. Given a Max 3-AND formula with $n'$ variables and $m'=\Delta n'$ clauses in which we want to distinguish between the case at most $((1-\beta)^3+\epsilon)m'$ clauses are satisfiable and the case that at least $(1-\frac{3}{2}\beta\diamond\gamma+\epsilon)m'$ clauses are satisfiable by $\gamma$-biased assignments, construct the following graph.

The left hand side (LHS) contains $2n'$ vertices, one for each literal. The right hand side (RHS) contains $m'$ clusters, one for each clause, where each cluster is a clique of size $m'$. In addition, the graph contains a clique of size $m''=(1-3\beta\diamond\gamma+\epsilon)m'^2$. In each cluster there is a unique vertex that is a "connecting vertex". Place an edge between a vertex that corresponds a literal and the connecting vertex of a cluster if the literal is in the clause that corresponds the cluster. These are called the "bipartite" edges.

In this graph, find a minimum bisection, which contains exactly $n'$ LHS vertices, and $(1-\frac{3}{2}\beta\diamond\gamma-\epsilon)m'$ clusters. It suffices to consider only the connecting vertices from each of the $m'$ clusters, and we need to find a cut of minimum width that contains $n'$ vertices from the LHS, and $(1-\frac{3}{2}\beta\diamond\gamma-\epsilon)m'$ connecting vertices.

When the 3-AND formula has $(1-\frac{3}{2}\beta\diamond\gamma-\epsilon)m'$ satisfiable clauses by $\gamma$-biased assignments, we pick the set $S$ to contain the clauses corresponding to these clauses and the $n'$ literals corresponding to the assignments consistent with these clauses. The only edges cut by this bisection connect the satisfying literals to unsatisfied clauses. The number of bipartite edges within the set $S$ is $3(1-\frac{3}{2}\beta\diamond\gamma-\epsilon)m'$. The sum of degrees of the satisfied literals is $3(1-\beta\diamond\gamma)m'$. Hence the width of the bisection is $\frac{3}{2}\beta\diamond\gamma+\epsilon$.

In a random 3-AND formula, we still need one side of the cut to contain $n'$ vertices and $(1-\frac{3}{2}\beta\diamond\gamma-\epsilon)m'$ clusters. This set of $n'$ literals has at most $((1-\beta\diamond\gamma)^3+\epsilon)m'$ of these clauses 3-connected to it (by Lemma 1) and the other $(1-\frac{3}{2}\beta\diamond\gamma-(1-\beta\diamond\gamma)^3-2\epsilon)m'$ clauses are 2-connected to it. Hence the width of the cut is at least

\begin{equation*}
\begin{split}
&3(1-\beta\diamond\gamma)m'-3((1-\beta\diamond\gamma)^3+\epsilon)m'-(1-\textstyle\frac{3}{2}\beta\diamond\gamma-(1-\beta\diamond\gamma)^3-2\epsilon)m'\\
&=(2(1-(1-\beta\diamond\gamma)^3)-\textstyle\frac{3}{2}\beta\diamond\gamma-\epsilon)m'\\
&\ge(2(1-(1-\beta)^3)-\frac{3}{2}\beta\diamond\gamma-\epsilon)m'.
\end{split}
\end{equation*}
\end{proof}

\begin{corollary}
Conjecture 2 holds for and arbitrarily small $\gamma$ and $\beta$ implies SSE is hard to approximate within $3-\epsilon$.
\end{corollary}

\section{Conjectures on Unbalanced k-CSP}
In this section, the author discusses the limitation of our method. Conjecture 1 can be generalized to that it is hard on average refuting Unbalanced Max k-CSP with Samorodnitsky-Trevisan hypergraph predicate under biased assignments on a natural distribution. The largest strengthened hardness of Min 2-Lin-2 that Conjecture 3 can yield is $2-\kappa$ where $\kappa$ is a small absolute positive. However, the author also shows that Conjecture 3 is not true for sufficiently large $k$. Hence, we cannot further strengthen the hardness for approximating Minimum Unique Game to $\omega_k(1)$, by proving that Min 2-Lin-2 is hard to approximate within any constant assuming Conjecture 3.

Let $k=2^r-1$. The Samorodnitsky-Trevisan hypergraph predicate\cite{st} of arity $k$ is the dual Hamming code (or truncated Hadamark code) $C$ of block length $k$ and dimension $r$ over $GF[2]$. If we index the position of a codeword $c=(c_S)_{\emptyset\ne S\subseteq[r]}$ by nonempty subsets $S$ of $[r]$, the codewords are given by $$C=\{c=(1-\sum_{i\in S}{y_i})_{\emptyset\ne S\subseteq[r]}|y_1,\cdots,y_r\in\mathbb{Z}_2\}.$$

Note that for every $c\in C$ and $c\ne(1,\cdots,1)$, the number of 1 in elements of $c$ is at least $(k+1)/2$.

Let $C$ be a k-ary predicate. In {\it Max $C$}, we are given a set of clauses, each clause contains exactly $k$ literals. A clause is satisfied if the values of literals satisfies $C$. The goal of the problem is to seek an assignment of the Boolean variables such that the number of satisfied clauses is maximized. We consider the case Max $C$ where $C$ is Samorodnitsiky-Trevisan hypergraph predicate and the case Max k-AND where $C$ is the predicate with one satisfying k-tuple $(1,\cdots,1)$, where $k=2^r-1$.

In {\it Random Unbalanced Max $C$}, we still assume that formulas are generated by the following random process. Given parameters $n$ and $m$, each clause is generated independently at random by selecting the $k$ variables in it independently at random and inserting the negative literal of the variable into the clause with probability $\beta<1/2$ and inserting the positive literal of the variable into the clause with probability $1-\beta$. $\beta$ is called {\it imbalance} of the instance, and the instance is called $\beta$-balanced. In addition, We are interested in the assignments such that the fraction of variables assigned to 0 is no more than $\gamma$, which is called {\it bias} of the assignments. We are also interested in the two cases: Max $C$ where $C$ is Samorodnitsky-Trevisan hypergraph predicate and Max k-AND where $C$ is the predicate with one satisfying k-tuple $(1,\cdots,1)$, where $k=2^r-1$.

In this section, the author considers the average complexity of Random Unbalanced Max $C$ with Samorodnitsky-Trevisan hypergraph predicate, and puts forward a variation of Feige's Hypothesis\cite{f,a}.\\

\noindent{\bf Conjecture 3. }
{\it Let $C$ be $k$-ary Samorodnitsky-Trevisan hypergraph predicate with $k=2^r-1$. For every $0<\gamma<\beta<1/2$, for every fixed $\epsilon>0$, for $\Delta$ a sufficiently large constant independent of $n$, there is no polynomial time algorithm that refutes most $\beta$-balanced Max $C$ formulas with $n$ variables and $m=\Delta n$ clauses, but never refutes a $1-\epsilon$ satisfiable formula under $\gamma$-biased assignments.}\\

The author also considers the average complexity of Random Unbalanced Max k-AND, and put forward the following conjecture.\\

\noindent{\bf Conjecture 4. }
{\it Suppose $k=2^r-1$. For every $0<\gamma<\beta<1/2$, for every fixed $\epsilon>0$, for $\Delta$ a sufficiently large constant independent of $n$, there is no polynomial time algorithm that refutes most $\beta$-balanced Max k-AND formulas with $n$ variables and $m=\Delta n$ clauses, but never refutes a $1-2\beta\diamond\gamma-\epsilon$ satisfiable formula under $\gamma$-biased assignments.}\\

We can prove that Conjecture 3 implies Conjecture 4 similarly as proof of Theorem 1.

\begin{theorem}
Conjecture 3 implies Conjecture 4.
\end{theorem}
\begin{proof}
We rewrite a formula of $\beta$-balanced Max $C$ to a formula of $\beta$-balanced Max k-AND. If the formula of Max $C$ is random, then the formula of Max 3-AND is also random. If the formula of Max $C$ $\phi$ is $1-\epsilon$ satisfiable by $\gamma$-biased assignments, we show in the following that at least $1-2\beta\diamond\gamma-\epsilon$ fraction of clauses in $\phi$ have all the $k$ literals satisfied.

On average, each positive literal has $k(1-\beta)\Delta$ appearance in $\phi$, and each negative literal has $k\beta\Delta$ appearance in $\phi$. When $\Delta$ is large enough, standard bounds on large deviations show that with high probability, all but an $\epsilon$ fraction of the occurrences of positive literals correspond to positive literals that appear between $(k(1-\beta)\pm\epsilon)\Delta$ times in $\phi$, and all but an $\epsilon$ fraction of the occurrences of negative literals correspond to negative literals that appear between $(k\beta\pm\epsilon)\Delta$ times in $\phi$.

Observe that every $\gamma$-biased assignment $\psi$ does not satisfy on average at most $k(\beta(1-\gamma)+\gamma(1-\beta))+\epsilon$ variables per clause in $\phi$. It then follows that at most $2(\beta(1-\gamma)+\gamma(1-\beta))+\epsilon$ clauses have at least $k/2$ literals unsatisfied by $\psi$.
\end{proof}

\begin{theorem}
Conjecture 4 holds for any $0<\gamma<\beta<1/2$ implies $(c',s')$-approximation hardness of Min UG for $c'=O(1/r)\beta\diamond\gamma$ and $s'=\Omega(1/k)(1-(1-\beta)^k)-o_k(\beta))$.
\end{theorem}
\begin{proof}
We use the $r$-dimensional hypercube gadget that is similar to the gadgets used by authors of \cite{h}.

Let $l_{1}\wedge\cdots\wedge l_{k}$ be a clause in the formula of Max k-AND, where $l_{i}$ is either a variable $x_{i}$ or its negation $\bar x_{i}$, for $i=1,\cdots,k$. The set of equations we construct have variables at the corners of a $r$-dimensional hypercube, which take value $1$ or $-1$. For each $\mu\in\{0,1\}^k$, we have a variable $v_{\mu}$. We let $u_{1},\cdots,u_{k}$ take the place of $v_{\mu}$, for $\mu$'s that are length-$r$ codes that have even number of 1. Let $u_{i}=-1$ if $x_{i}=1$, and $u_{i}=1$ if $x_{i}=0$. For each edge $(u_{i},v_{\mu})$ of the cube, we have the equation $u_{i}v_{\mu}=1$ if $l_{i}$ is positive, and the equation $u_{i}v_{\mu}=-1$ if $l_{i}$ is negative, for all $i=1,\cdots,k$.

If all $l_{i}$ are satisfied in the clause, we assign $v_{\mu}$ the value $(-1)^{\mu_{1}+\cdots+\mu_{k}}$. All the edge equations are satisfied and left no equation unsatisfied. Otherwise, it is only possible to satisfy at most $1-\Omega(1/k)$ fraction of equations and left $\Omega(1/k)$ fraction of equations unsatisfied, and that it is always possible to satisfy at least $1-O(1/r)$ equations and left $O(1/r)$ equations unsatisfied.

Given a $\beta$-balanced Max k-AND formulas $\phi$ that is at most $(1-\beta)^k+o_k(\beta)$ satisfiable, at least $1-(1-\beta)^k-o_k(\beta)$ clauses in $\phi$ are unsatisfied.

Now we reduce a formula of $\beta$-balanced Max k-AND to an instance of Min 2-Lin-2 using the gadget introduced above. If the formula is $1-2\beta\diamond\gamma-\epsilon$ satisfiable under $\gamma$-biased assignments, then the value of the instance of Min 2-Lin-2 is at most $$2\beta\diamond\gamma+\epsilon-(2\beta\diamond\gamma+\epsilon)(1-O(1/r))=O(1/r)\beta\diamond\gamma+\epsilon.$$ If the formula is random, then it is at most $(1-\beta)^k+o_k(\beta)$ satisfiable in high probability, which implies the value of the instance of Min 2-Lin-2 is at least $$1-(1-\beta)^k-o_k(\beta)-(1-(1-\beta)^k-o_k(\beta))(1-\Omega(1/k))=\Omega(1/k)(1-(1-\beta)^k)-o_k(\beta).$$
\end{proof}

\begin{theorem}
Conjecture 2 holds for any $0<\gamma<\beta<1/2$ implies Small Set Expansion is hard to approximate within $\frac{k-1}{k}\frac{1-(1-\beta)^k}{\beta\diamond\gamma}-\frac{k-2}{k}-\epsilon$.
\end{theorem}
\begin{proof}
We reduce $\beta$-balanced Max k-AND to Min Bisection. Given a Max k-AND formula with $n'$ variables and $m'=\Delta n'$ clauses in which we want to distinguish between the case at most $((1-\beta)^k+\epsilon)m'$ clauses are satisfiable and the case that at least $(1-2\beta\diamond\gamma+\epsilon)m'$ clauses are satisfiable by $\gamma$-biased assignments, construct the following graph.

The left hand side (LHS) contains $2n'$ vertices, one for each literal. The right hand side (RHS) contains $m'$ clusters, one for each clause, where each cluster is a clique of size $m'$. In addition, the graph contains a clique of size $m''=(1-4\beta\diamond\gamma+\epsilon)m'^2$. In each cluster there is a unique vertex that is a "connecting vertex". Place an edge between a vertex that corresponds a literal and the connecting vertex of a cluster if the literal is in the clause that corresponds the cluster. These are called the "bipartite" edges.

In this graph, find a minimum bisection, which contains exactly $n'$ LHS vertices, and $(1-2\beta\diamond\gamma-\epsilon)m'$ clusters. It suffices to consider only the connecting vertices from each of the $m'$ clusters, and we need to find a cut of minimum width that contains $n'$ vertices from the LHS, and $(1-2\beta\diamond\gamma-\epsilon)m'$ connecting vertices.

When the k-AND formula has $(1-2\beta\diamond\gamma-\epsilon)m'$ satisfiable clauses by $\gamma$-biased assignments, we pick the set $S$ to contain the clauses corresponding to these clauses and the $n'$ literals corresponding to the assignments consistent with these clauses. The only edges cut by this bisection connect the satisfying literals to unsatisfied clauses. The number of bipartite edges within the set $S$ is $k(1-2\beta\diamond\gamma-\epsilon)m'$. The sum of degrees of the satisfied literals is $k(1-\beta\diamond\gamma)m'$. Hence the width of the bisection is $k\beta\diamond\gamma+\epsilon$.

In a random k-AND formula, we still need one side of the cut to contain $n'$ vertices and $(1-2\beta\diamond\gamma-\epsilon)m'$ clusters. This set of $n'$ literals has at most $((1-\beta\diamond\gamma)^k+\epsilon)m'$ of these clauses $k$-connected to it (by Lemma 1) and the other $(1-2\beta\diamond\gamma-(1-\beta\diamond\gamma)^k-2\epsilon)m'$ clauses are $(k-1)$-connected to it. Hence the width of the cut is at least

\begin{equation*}
\begin{split}
&k(1-\beta\diamond\gamma)m'-k((1-\beta\diamond\gamma)^k+\epsilon)m'-(1-2\beta\diamond\gamma-(1-\beta\diamond\gamma)^k-2\epsilon)m'\\
&=((k-1)(1-(1-\beta\diamond\gamma)^k)-(k-2)\beta\diamond\gamma-\epsilon)m'\\
&\ge ((k-1)(1-(1-\beta)^k)-(k-2)\beta\diamond\gamma-\epsilon)m'.
\end{split}
\end{equation*}
\end{proof}

\section{Discussion}
Notice that for Theorem 5 to make sense, we have $\beta=O(r/k)$. However, by the construction of gadgets in proof of Theorem 5, we can reduce the instance of Min 2-Lin-2 to an instance of Min 2-SAT, where at least $1-O(\beta^2)$ fraction of the clauses is Horn. Since Min Horn-2-SAT can be approximated within 2 by a LP algorithms\cite{gz}, Conjecture 3 cannot be true when $k$ is so large so that the hardness exceeds 2. The strongest hardness result of Min 2-Lin-2 that Conjecture 3 yields is $2-\kappa$, where $\kappa$ is a small absolute positive.

In Fig. 1, the dark gray area is the known region of $(c,s)$-approximation NP-hardness of UG, the light gray area at the top left corner is the region of $(c,s)$-approximation hardness of UG assuming Conjecture 3 for certain $k$.

\begin{figure}
\begin{center}
\includegraphics[width=0.8\textwidth,bb=87 262 507 578]{illustration.eps}
\caption{illustration of $(c,s)$-approximation of UG}
\end{center}
\end{figure}

\end{document}